\RequirePackage{fix-cm}
\documentclass[natbib,smallextended]{svjour3}       
\smartqed  
\usepackage{graphicx}
\usepackage{nicefrac}

\usepackage{lineno,hyperref}

\usepackage[utf8]{inputenc}
\usepackage{amsmath}
\usepackage{graphicx}
\usepackage{array}
\usepackage[colorinlistoftodos]{todonotes}
\usepackage{color}
\usepackage{framed}
\usepackage{hyperref}
\usepackage{paralist}
\usepackage{pdfpages}

\usepackage{tikz}
  \usetikzlibrary{shapes.geometric}
\usetikzlibrary{arrows.meta,arrows}

\usepackage[left=3cm,top=3cm,right=3cm,bottom=3cm,bindingoffset=0.5cm]{geometry}

\usepackage{soul}
\usepackage{nicefrac}

\definecolor{fabio}{rgb}{0,0.3,0.7}
\newcommand{\fabio}[1]{\textcolor{fabio}{#1}}

\newcommand{\marginfab}[1]{\marginpar{\tiny{\fabio{#1}}}}

\newcommand{\steven}[1]{\textcolor{black}{#1}}

\newtheorem{observation}{Observation}

\begin{document}

\title{Finding the most parsimonious or likely tree in a network with respect to an alignment
\thanks{Kelk and Pardi contributed equally to this article.}
}
%


\titlerunning{Finding the most parsimonious or likely tree in a network with respect to an alignment}        

\author{Steven Kelk         \and
        Fabio Pardi \and Celine Scornavacca \and Leo van Iersel 
}


\institute{S. Kelk \at
Department of Data Science and Knowledge Engineering (DKE), Maastricht University, P.O. Box 616, 6200 MD Maastricht, Netherlands.\\
              \email{steven.kelk@maastrichtuniversity.nl}           
           \and
	   F. Pardi \at
  Laboratoire d’Informatique, de Robotique et de
Micro\'{e}lectronique de Montpellier (LIRMM),
Universit\'{e} de Montpellier, CNRS, Institut de Biologie Computationnelle (IBC),
34095 Montpellier Cedex 5, France.\\
              \email{pardi@lirmm.fr}           
           \and
	   C. Scornavacca \at
Institut des Sciences de l'Evolution, 
Universit\'{e} de Montpellier, CNRS, IRD, EPHE\\
Institut de Biologie Computationnelle (IBC),
34095 Montpellier Cedex 5, France.\\
              \email{Celine.Scornavacca@umontpellier.fr}           
           \and
	  L. van Iersel \at
Delft Institute of Applied Mathematics, Delft University of Technology, Mekelweg 4,
2628 CD, Delft, Netherlands.\\
              \email{l.j.j.v.iersel@gmail.com}
}

\date{Received: date / Accepted: date}

\maketitle

\begin{abstract}
Phylogenetic networks are often constructed by merging multiple conflicting phylogenetic signals into a directed acyclic graph. It is interesting to explore whether a network constructed in this way induces biologically-relevant phylogenetic signals that were not present in the input.

Here we show that, given a multiple alignment $A$ for a set of taxa $X$ and a rooted phylogenetic network $N$ whose leaves are labelled by $X$, it is NP-hard to locate the most parsimonious phylogenetic tree displayed by $N$ (with respect to $A$) even when the level of $N$ - the maximum number of reticulation nodes within a biconnected component - is 
1 and $A$ contains only 2 distinct states. \steven{(If, additionally, gaps are allowed the problem becomes APX-hard.)} \steven{We also show that under the same conditions, and assuming 
a simple binary symmetric model of character evolution,
finding the most likely tree displayed by the network is NP-hard.} 

These negative results contrast with earlier work on parsimony in which it is shown that if $A$ consists of a single column the problem is fixed parameter tractable in the level. 
We conclude with a discussion of why, despite \steven{the NP-hardness, both the parsimony and likelihood problem can likely be well-solved in practice.}

\end{abstract}

\section{Introduction}
Rooted phylogenetic networks are generalizations of rooted phylogenetic trees which allow horizontal evolutionary events such as horizontal gene transfer, recombination and hybridization to be modelled (\cite{Huson2010,Morrison2011}). This is achieved by allowing nodes with indegree 2 or higher, known as \emph{reticulation} nodes. Recent years have seen an explosion of interest in constructing rooted phylogenetic networks, fuelled by the growing awareness that \emph{incongruence} in phylogenetic and phylogenomic datasets is not simply a question of evolutionary ``noise'', but sometimes the result of evolutionary phenomena more complex than speciation and mutation (e.g. \cite{Zhaxybayeva2011,Abbott2013,Vuilleumier2015}).

Although many modelling questions surrounding the construction of phylogenetic networks are still to be answered, it is commonplace to associate a rooted phylogenetic network with the set of rooted phylogenetic trees that it contains (``displays''). Informally speaking, a rooted phylogenetic network displays a rooted phylogenetic tree if the tree can be topologically embedded inside the network. A network is not necessarily defined by the set of trees it displays,
but the notion of display is nevertheless a recurring theme in the literature, since networks themselves are often constructed by merging phylogenetic trees subject to some optimality criterion.
It is well-known that it is NP-hard to determine whether a network displays a given tree (\cite{kanj2008seeing}), although on many restricted classes of phylogenetic networks the problem is polynomial-time solvable (\cite{van2010locating,fakcharoenphol2015faster,gambette2015locating}).

Although it is important to be able to determine whether a network displays a given tree, we may also wish to ask what the ``best'' tree is within the network, subject to some optimality criterion. For example, given a network $N$ and a multiple alignment $A$, we may wish to ask for the single tree $T$ displayed by $N$ with lowest parsimony score with respect to $A$. Similarly, if the network is decorated by edge lengths or probabilities, we may wish to identify the most \emph{likely} tree displayed by the network, that is, the tree that maximizes the probability 
of generating $A$, under a given model of evolution. 
Such questions are natural, as the two following examples show. First, phylogenetic networks are often constructed by topologically merging incongruent phylogenetic signals (e.g. \cite{kelk2014constructing}), and it is insightful to ask whether the network thus constructed displays trees which have interesting properties (such as a low parsimony score or a high likelihood) which were not in the input. Second, we may wish to perform classical phylogenetic tree construction under criteria such as maximum parsimony or maximum likelihood (e.g. \cite{jin2006maximum,jin2007efficient}), but within the restricted space of trees displayed by a given network.

Problems of the above kind are already known to be NP-hard, since it is NP-hard to determine the most parsimonious tree $T$ displayed by a given network $N$ even when the alignment $A$ consists of a \emph{single} column and the network is binary (\cite{fischer2015computing}). However, the gadgets used in that hardness reduction produce networks with very high \emph{level}, where 
level is the maximum number of reticulation nodes in a biconnected component of the network. On the positive side, the same article shows that the problem on an alignment with a single column is FPT in the level of the network. This means that, on a network with level $k$, the problem can be solved in time $f(k) \cdot \text{poly}(n)$ where $f$ is a function that depends only on $k$ and $n$ is the size of the network.
Such results are useful in practice, when (as is often the case) $k$ is small.

The question emerges whether the positive FPT result goes through when $A$ does not consist of a single column, but potentially many columns -- a problem introduced more than one decade ago (\cite{nakhleh2005reconstructing}). Here we show that this is not the case. We prove the rather negative result that locating the most parsimonious tree in a rooted binary network $N$ is NP-hard, even under the following restricted circumstances: (1) each biconnected component of the network contains exactly one reticulation node (i.e. is ``level-1''); (2) each biconnected component of the network has exactly three outgoing arcs; (3) the alignment $A$ consists of two states. If indel symbols are permitted then the problem is not only NP-hard, but also difficult to approximate well (APX-hard). If any of the conditions (1)-(3) are further strengthened (respectively: the network becomes a tree; the reticulation nodes become redundant; the alignment becomes uninformative), the problem becomes trivially solvable, so in some sense this is a ``best possible'' (or the ``worst possible'', depending on your perspective) hardness result. Next, we consider the question of identifying the most likely tree in the network. We obtain NP-hardness under the same restrictions (1)-(3), 
subject to the simple binary symmetric model of character evolution.
It is no coincidence that restrictions (1)-(3) again apply, since the hardness of the likelihood question is established by reducing the parsimony variant of the problem to it. Specifically, we show that the most likely tree displayed by 
a network with sufficiently short branches
is necessarily also a most parsimonious tree. 

Although the main results in this paper are negative, some reasons for hope are given in the conclusion.

 
\section{Preliminaries}

A \emph{rooted binary phylogenetic network} $N$ on a set $X$ of taxa is a directed acyclic graph where the leaves (nodes of indegree-1 and outdegree-0) are bijectively labelled by $X$, there is a unique root (a node of indegree-0 and outdegree-2) and all other nodes are either tree nodes (indegree-1 and outdegree-2) or reticulation nodes (indegree-2 and outdegree-1). For brevity we henceforth simply use the term \emph{network}. A \emph{rooted binary phylogenetic tree} (henceforth \emph{tree}) is a phylogenetic network without any reticulation nodes. A \emph{cherry} is a pair of taxa that share a common parent. A \emph{rooted binary caterpillar} is a tree with exactly one cherry.

The \emph{level} of a network $N$ is the maximum number of reticulation nodes in a biconnected component of the undirected graph underpinning $N$. In this article we will focus exclusively on level-1 networks. In level-1 networks, maximal biconnected components that are not single edges are simple cycles that contain exactly one reticulation node; such biconnected components are  called \emph{galls}. An arc whose tail (but not head) is a node of a gall is called an \emph{outgoing} arc.

A \emph{character} $f$ is a surjective mapping
$f: X \rightarrow S$ where $S$ is a set of discrete \emph{states}. When $S$ contains two 
states we say that $f$ is a \emph{binary} character. Given a tree $T = (V,E)$ and a character $f$, both on $X$, we say that $\hat{f}: V \rightarrow S$ is an \emph{extension} of $f$ to $T$ if
$\hat{f}(x) = f(x)$ for all $x \in X$. The \emph{number of mutations} induced by $\hat{f}$ (on $T$), denoted $l_{\hat{f}}(T)$, 
is the number of edges $\{u,v\} \in E$ such that $\hat{f}(u) \neq \hat{f}(v)$. The \emph{parsimony score} of $f$ with respect to $T$, denoted $l_f(T)$, is the minimum number of mutations induced ranging over all extensions $\hat{f}$ of $f$. Any extension that achieves
this minimum is called an \emph{optimal} extension. An optimal extension can be computed in polynomial time using Fitch's algorithm (\cite{fitch1971}), which for completeness we describe in the appendix along with some of its relevant mathematical properties. (Note that there potentially exist optimal extensions that cannot be generated by Fitch's algorithm.)

For a network $N$ and a tree $T$, both on $X$, we say that $N$ \emph{displays} $T$ if there exists a subtree $T'$ of $N$ such that $T'$ is a subdivision of $T$.  An equivalent definition of ``displays'' relies on the notion of a \emph{switching}, where a switching is a subtree $N'$ of $N$ obtained by, for each reticulation node $u$, deleting exactly one of $u$'s incoming edges. $N$ displays $T$ if and only if there exists some switching $N'$ of $N$ and a subdivision $T'$ of $T$ such that $T'$ is a subgraph of $N'$. In both definitions we say that $T'$ is an \emph{image} of $T$ inside $N$.

The softwired parsimony score\footnote{Two other definitions of the parsimony score of a network exist in the literature: the \emph{hardwired} parsimony score (\cite{kannan2012maximum}) and parental parsimony score (\cite{vanIersel2017}); see the latter manuscript for a discussion about the differences of these three models.} of a network $N$ with respect to $f$ is the minimum, ranging over all trees $T$  displayed by $N$, of $l_f(T)$.

We now extend the above concepts to \emph{alignments}. An alignment $A$ is simply 
a linear ordering of characters. In this paper the linear ordering is irrelevant so we can arbitrarily impose an ordering and write $f \in A$ without ambiguity. An alignment can naturally be represented as a matrix with $|X|$ rows and $|A|$ columns; we therefore use the terms characters and columns interchangeably (and, following the use of alignments in practice, we sometimes refer to the rows of the matrix as \emph{sequences}). The parsimony score of a tree $T$ with respect to $A$, denoted $l_{A}(T)$, is simply $\sum_{f \in A}l_f(T)$. 

When extending this concept to networks, two definitions have been proposed: the parsimony score of a network  with respect to an alignment  $A$, denoted  $l_{A}(N)$, can be defined as
\begin{enumerate}
\item $\displaystyle\sum_{f \in A}
\displaystyle\min_{T \in\mathcal{T}(N)}l_f(T)$
\item[or]
\item $\displaystyle\min_{T \in\mathcal{T}(N)}\sum_{f \in A}l_f(T)$
\end{enumerate}
where $\mathcal{T}(N)$ is the set of trees displayed by the network. According to the first definition (introduced in \cite{hein1990reconstructing}), each character can follow a different tree displayed by the network, while in the second one (introduced in \cite{nakhleh2005reconstructing}) all characters of the alignment follow the same tree. In this paper, we will adopt the latter definition, and a
tree $T$ that is the
minimizer of this sum is called the \emph{most parsimonious (MP) tree displayed by $N$} (with respect to $A$).

Note that in applied phylogenetics alignments often contain indels, encoded using a gap symbol ``-''. From the parsimony perspective it is not uncommon to treat these symbols as wildcards that  do not induce mutations; the taxon ``does not care'' what state it is assigned. (Note however that extensions are \emph{not} allowed to contain gap symbols). To compute $l_f(T)$ when a character $f: X \rightarrow S$ maps some of its taxa to the gap symbol, we can run Fitch's algorithm with a slight modification to the bottom-up phase: for each taxon $x$ such that $f(x) = $``-'', we assign the entire set of states
$S$ to $x$.
Moreover, as the following observation shows, the use of ``-'' symbols does not make the problem of identifying the most parsimonious tree displayed by a network significantly harder.

\begin{observation}
\label{obs:recode}
Let $A$ be an alignment for a set of taxa $X$ and let $N$ be a phylogenetic network on $X$. Suppose $A$ uses the states $\{0,1,\text{``-''}\}$.  Let $k$ denote the total number of gap symbols in $A$. 
In polynomial time we can construct an alignment $A'$ on $2|X|$ taxa, which uses only states $\{0,1\}$, and a network $N'$ on $2|X|$ taxa, such that there is a polynomial-time computable bjiection $g$ mapping trees displayed by $N$ to trees displayed by $N'$. This bijection $g$ has the property that, for each tree $T$ displayed by $N$, $l_{A'}( g(T) )= l_A(T) + k$. Consequently, $T$ is a most parsimonious tree displayed by $N$ (wrt $A$) if and only if 
$g(T)$ is a most parsimonious tree displayed by $N'$ (wrt $A'$).
\end{observation}
\begin{proof}
To obtain $N'$ from $N$ we split each taxon $x_i$ into a cherry $\{x^1_i$, $x^2_i\}$. If, for a given character, $x_i$ had state $0$ (respectively, $1$), we give both $x^1_i$ and $x^2_i$ the
state $0$ (respectively, $1$). If $x_i$ had state ``-'' we give $x^1_i$ state 0 and $x^2_i$ state 1. The idea is that by encoding a gap symbol as a $\{0,1\}$ cherry a single mutation is unavoidably incurred (on one of the two edges leading into $x^1_i$ and $x^2_i$) and thus the state of the parent of the cherry in any (optimal) extension is irrelevant. The parent thus simulates the original gap symbol: the bottom-up phase of Fitch's algorithm will always allocate the subset of states $\{0,1\}$ to the parent. (The bijection $g$, and its inverse, are trivially
computable in polynomial time by splitting each taxon into a cherry, or collapsing cherries,
respectively). \qed
\end{proof}

We defer preliminaries relating to likelihood until Section \ref{sec:MLtree}.

Let $G$ be an undirected graph. An \emph{orientation} of $G$ is a directed graph $G'$ obtained by replacing each edge $\{u,v\}$ of $G$ with exactly one of the two arcs $(u,v)$ or $(v,u)$. Given an orientation $G'$ of $G$, a \emph{source} is a node that has only outgoing arcs, and a \emph{sink} is a node that has only incoming arcs. Let $msso(G)$ denote the maximum, ranging over all possible orientations $G'$ of $G$, of the sum of the number of sources and sinks in $G'$. MAX-SOURCE-SINKS-ORIENTATION is the problem of computing $msso(G)$. A cubic graph is a graph where every node has degree 3.\\
\\
The proofs of the following are deferred to the appendix. These two results form the foundation of the hardness results given in the next section.\\

\begin{lemma}
\label{lem:sourcesinkhard}
MAX-SOURCE-SINKS-ORIENTATION is NP-hard on cubic graphs.
\end{lemma}

\begin{corollary}
\label{cor:sourcesinksapxhard}
MAX-SOURCE-SINKS-ORIENTATION is APX-hard on cubic graphs.
\end{corollary}

\section{Hardness of finding the most parsimonious tree displayed by a network}
\label{sec:hard}

In this section we will build on Lemma \ref{lem:sourcesinkhard} and Corollary \ref{cor:sourcesinksapxhard} to prove that 
 computing the most parsimonious tree displayed by a rooted phylogenetic network $N$ with respect to an alignment $A$ is  NP-hard  and APX-hard already for highly restricted instances.

\begin{theorem} \label{thm:NPhardness_MP}
It is NP-hard to compute the most parsimonious tree displayed by a rooted phylogenetic network $N$ with respect to an alignment $A$, even when $N$ is a binary level-1 network with at most 3 outgoing arcs per gall and $A$ consists 
only of two states $\{0,1\}$ and does not contain 
gap symbols.
\end{theorem}
\begin{proof}

Let $G = (V,E)$ be a cubic instance of MAX-SOURCE-SINKS-ORIENTATION. We will start by building
a binary level-1 network $N$ with $6|E|$ taxa and $2|E|$ reticulations, and an alignment $A$ on states $\{0,1,\text{``-''}\}$ consisting of $6|E|$ sequences, each sequence of length $|V|$. (We  will remove the ``-'' symbols later). One can thus view $A$ as a $\{0,1,-\}$ matrix with $6|E|$ rows and $|V|$ columns, or
equivalently as a set of $|V|$ characters for the $6|E|$ taxa of $N$.

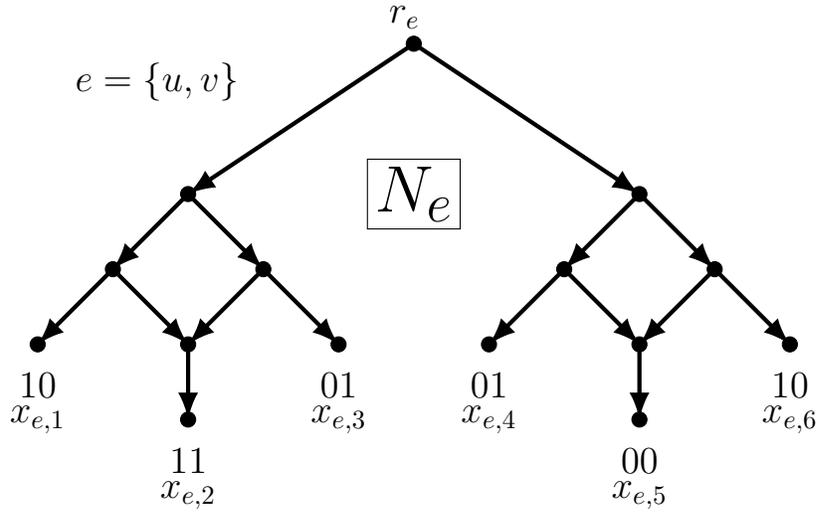
\begin{figure}
\centering
\begin{tikzpicture}[scale=1]
\begin{scope}[xshift=0cm,yshift=0cm]
\draw [fill] (-3,-2) circle (0.1);
\draw [fill] (0,0) circle (0.1) node[above,text width=1cm,align=center] {\Large{$r_e$\vspace{1mm}\phantom{0}}};
\draw [fill] (3,-2) circle (0.1);
\draw[-Latex,ultra thick] (0,0) to (3,-2);
\draw[-Latex,ultra thick] (0,0) to (-3,-2);
\draw [fill] (-4,-3) circle (0.1);
\draw [fill] (-2,-3) circle (0.1);
\draw [fill] (4,-3) circle (0.1);
\draw [fill] (2,-3) circle (0.1);
\draw [fill] (-3,-4) circle (0.1);
\draw [fill] (3,-4) circle (0.1);
\draw[-Latex,ultra thick] (-3,-2) to (-4,-3);
\draw[-Latex,ultra thick] (-3,-2) to (-2,-3);
\draw[-Latex,ultra thick] (3,-2) to (4,-3);
\draw[-Latex,ultra thick] (3,-2) to (2,-3);
\draw[-Latex,ultra thick] (-4,-3) to (-3,-4);
\draw[-Latex,ultra thick] (-2,-3) to (-3,-4);
\draw[-Latex,ultra thick] (4,-3) to (3,-4);
\draw[-Latex,ultra thick] (2,-3) to (3,-4);
\draw [fill] (-5,-4) circle (0.1) node[below,text width=1cm,align=center] {\Large{\\10\\$x_{e,1}$}};
\draw [fill] (-1,-4) circle (0.1) node[below,text width=1cm,align=center] {\Large{\\01\\$x_{e,3}$}};
\draw [fill] (1,-4) circle (0.1) node[below,text width=1cm,align=center] {\Large{\\01\\$x_{e,4}$}};
\draw [fill] (5,-4) circle (0.1) node[below,text width=1cm,align=center] {\Large{\\10\\$x_{e,6}$}};
\draw [fill] (-3,-5) circle (0.1) node[below,text width=1cm,align=center] {\Large{\\11\\$x_{e,2}$}};
\draw [fill] (3,-5) circle (0.1) node[below,text width=1cm,align=center] {\Large{\\00\\$x_{e,5}$}};
\draw[-Latex,ultra thick] (-4,-3) to (-5,-4);
\draw[-Latex,ultra thick] (-2,-3) to (-1,-4);
\draw[-Latex,ultra thick] (-3,-4) to (-3,-5);
\draw[-Latex,ultra thick] (2,-3) to (1,-4);
\draw[-Latex,ultra thick] (4,-3) to (5,-4);
\draw[-Latex,ultra thick] (3,-4) to (3,-5);
\node[draw] at (0,-2) {\Huge{$N_e$}};
\node[text width=3cm] at (-3,-0.5)  {\Large{$e=\{u,v\}$}};
\end{scope}
\end{tikzpicture}
\caption{Although each sequence has length $|V|$, only columns $u$ and $v$ are shown. For
this edge, the other $|V|-2$ symbols are ``-''.\label{fig:stevenfig1}}
\end{figure}

To construct $N$, we start by taking a rooted binary caterpillar on $|E|$ taxa. For each $e \in E$ replace the taxon $x_e$ of the caterpillar with a copy $N_e$ of the network shown in Figure \ref{fig:stevenfig1}. The 6 taxa within $N_e$ are denoted
$x_{e,i}$, $i \in \{1,\ldots,6\}$.  We use $r_e$ to refer to the root of $N_e$.

To construct the alignment, we write $A_{e,i}$ $(e \in E, i \in \{1, \ldots 6\})$ to refer to
the sequences, and write $A_{e,i,v}$ to refer to the state in its $v$th column.
These states are assigned as follows.
For each edge $e = \{u,v\} \in E$, we set the states of the 6 taxa $A_{e,i,u}$ ($i \in \{1, \ldots, 6\}$)
to be $1, 1, 0, 0, 0, 1$, the
states of the 6 taxa $A_{e,i,v}$ ($i \in \{1, \ldots, 6\}$) to be
$0, 1, 1, 1, 0, 0$, and for each $w \not \in \{u,v\}$, we set
the states of the 6 taxa $A_{e,i,w}$ ($i \in \{1, \ldots, 6\}$) to all be ``-''. Given
that each edge is incident to exactly 3 edges, there are exactly $k := 6|V|(|E|-3)$ ``-''
symbols in $A$.

\begin{figure}
\centering
\begin{tikzpicture}[scale=0.5]
\begin{scope}[xshift=0cm,yshift=0cm]
\node[text width=0.2cm] at (-3.1,-1.4)  {$\cup$};
\node[text width=0.2cm] at (3.1,-1.4)  {$\cup$};
\node[text width=0.2cm] at (-3.4,-3.2)  {$\cup$};
\node[text width=0.2cm] at (2.6,-3.2)  {$\cup$};
\draw [fill] (-3,-2) circle (0.1) node[left,align=center]{\small{\{0,1\}1}};
\draw [fill] (0,0) circle (0.1) node[above,text width=2cm,align=center]{\small{$\cup$\\\{0,1\}\{0,1\}}};
\draw [fill] (3,-2) circle (0.1) node[right,align=center]{\small{\{0,1\}0}};
\draw[-Latex,thick] (0,0) to (3,-2);
\draw[-Latex,thick] (0,0) to (-3,-2);
\draw [fill] (-4,-3) circle (0.1)  node[left,align=center]{\small{1\{0,1\}}};
\draw [fill] (2,-3) circle (0.1) node[left,align=center]{\small{0\{0,1\}}};
\draw[-Latex,thick] (-3,-2) to (-4,-3);
\draw[-Latex,thick] (-3,-2) to (-1,-4);
\draw[-Latex,thick] (3,-2) to (5,-4);
\draw[-Latex,thick] (3,-2) to (2,-3);
\draw[-Latex,thick] (-4,-3) to (-3,-5);
\draw[-Latex,thick] (2,-3) to (3,-5);
\draw [fill] (-5,-4) circle (0.1) node[below,text width=1cm,align=center] {\\10\\$x_{e,1}$};
\draw [fill] (-1,-4) circle (0.1) node[below,text width=1cm,align=center] {\\01\\$x_{e,3}$};
\draw [fill] (1,-4) circle (0.1) node[below,text width=1cm,align=center] {\\01\\$x_{e,4}$};
\draw [fill] (5,-4) circle (0.1) node[below,text width=1cm,align=center] {\\10\\$x_{e,6}$};
\draw [fill] (-3,-5) circle (0.1) node[below,text width=1cm,align=center] {\\11\\$x_{e,2}$};
\draw [fill] (3,-5) circle (0.1) node[below,text width=1cm,align=center] {\\00\\$x_{e,5}$};
\draw[-Latex,thick] (-4,-3) to (-5,-4);
\draw[-Latex,thick] (2,-3) to (1,-4);
\node[draw,text width=1.9cm] at (-4,0.5)  {Switching 1\\5 mutations};
\end{scope}
\end{tikzpicture}
\begin{tikzpicture}[scale=0.5]
\begin{scope}[xshift=0cm,yshift=0cm]
\node[text width=0.2cm] at (-3.1,-1.4)  {$\cup$};
\node[text width=0.2cm] at (3.1,-1.4)  {$\cup$};
\node[text width=0.2cm] at (-3.4,-3.2)  {$\cup$};
\node[text width=0.2cm] at (3.4,-3.2)  {$\cup$};
\draw [fill] (-3,-2) circle (0.1) node[left,align=center]{\small{\{0,1\}1}};
\draw [fill] (0,0) circle (0.1) node[above,text width=2cm,align=center]{\small{01}};
\draw [fill] (3,-2) circle (0.1) node[right,align=center]{\small{0\{0,1\}}};
\draw[-Latex,thick] (0,0) to (3,-2);
\draw[-Latex,thick] (0,0) to (-3,-2);
\draw [fill] (-4,-3) circle (0.1)  node[left,align=center]{\small{1\{0,1\}}};
\draw [fill] (4,-3) circle (0.1) node[right,align=center]{\small{\{0,1\}0}};
\draw[-Latex,thick] (-3,-2) to (-4,-3);
\draw[-Latex,thick] (-3,-2) to (-1,-4);
\draw[-Latex,thick] (3,-2) to (5,-4);
\draw[-Latex,thick] (3,-2) to (1,-4);
\draw[-Latex,thick] (-4,-3) to (-3,-5);
\draw [fill] (-5,-4) circle (0.1) node[below,text width=1cm,align=center] {\\10\\$x_{e,1}$};
\draw [fill] (-1,-4) circle (0.1) node[below,text width=1cm,align=center] {\\01\\$x_{e,3}$};
\draw [fill] (1,-4) circle (0.1) node[below,text width=1cm,align=center] {\\01\\$x_{e,4}$};
\draw [fill] (5,-4) circle (0.1) node[below,text width=1cm,align=center] {\\10\\$x_{e,6}$};
\draw [fill] (-3,-5) circle (0.1) node[below,text width=1cm,align=center] {\\11\\$x_{e,2}$};
\draw [fill] (3,-5) circle (0.1) node[below,text width=1cm,align=center] {\\00\\$x_{e,5}$};
\draw[-Latex,thick] (-4,-3) to (-5,-4);
\draw[-Latex,thick] (4,-3) to (3,-5);
\node[draw,text width=1.9cm] at (-4,0.5)  {Switching 2\\4 mutations};
\end{scope}
\end{tikzpicture}
\\
\begin{tikzpicture}[scale=0.5]
\begin{scope}[xshift=0cm,yshift=0cm]
\node[text width=0.2cm] at (-3.1,-1.4)  {$\cup$};
\node[text width=0.2cm] at (3.1,-1.4)  {$\cup$};
\node[text width=0.2cm] at (3.4,-3.2)  {$\cup$};
\node[text width=0.2cm] at (-2.7,-3.2)  {$\cup$};
\draw [fill] (-3,-2) circle (0.1) node[left,align=center]{\small{1\{0,1\}}};
\draw [fill] (0,0) circle (0.1) node[above,text width=2cm,align=center]{\small{$\cup$\\\{0,1\}\{0,1\}}};
\draw [fill] (3,-2) circle (0.1) node[right,align=center]{\small{0\{0,1\}}};
\draw[-Latex,thick] (0,0) to (3,-2);
\draw[-Latex,thick] (0,0) to (-3,-2);
\draw [fill] (-2,-3) circle (0.1)  node[right,align=center]{\small{\{0,1\}1}};
\draw [fill] (4,-3) circle (0.1) node[right,align=center]{\small{\{0,1\}0}};
\draw[-Latex,thick] (-3,-2) to (-5,-4);
\draw[-Latex,thick] (-3,-2) to (-1,-4);
\draw[-Latex,thick] (3,-2) to (5,-4);
\draw[-Latex,thick] (3,-2) to (1,-4);
\draw[-Latex,thick] (-2,-3) to (-3,-5);
\draw [fill] (-5,-4) circle (0.1) node[below,text width=1cm,align=center] {\\10\\$x_{e,1}$};
\draw [fill] (-1,-4) circle (0.1) node[below,text width=1cm,align=center] {\\01\\$x_{e,3}$};
\draw [fill] (1,-4) circle (0.1) node[below,text width=1cm,align=center] {\\01\\$x_{e,4}$};
\draw [fill] (5,-4) circle (0.1) node[below,text width=1cm,align=center] {\\10\\$x_{e,6}$};
\draw [fill] (-3,-5) circle (0.1) node[below,text width=1cm,align=center] {\\11\\$x_{e,2}$};
\draw [fill] (3,-5) circle (0.1) node[below,text width=1cm,align=center] {\\00\\$x_{e,5}$};
\draw[-Latex,thick] (4,-3) to (3,-5);
\node[draw,text width=1.9cm] at (-4,0.5)  {Switching 3\\5 mutations};
\end{scope}
\end{tikzpicture}
\begin{tikzpicture}[scale=0.5]
\begin{scope}[xshift=0cm,yshift=0cm]
\node[text width=0.2cm] at (-3.1,-1.4)  {$\cup$};
\node[text width=0.2cm] at (3.1,-1.4)  {$\cup$};
\node[text width=0.2cm] at (2.6,-3.2)  {$\cup$};
\node[text width=0.2cm] at (-2.7,-3.2)  {$\cup$};
\draw [fill] (-3,-2) circle (0.1) node[left,align=center]{\small{1\{0,1\}}};
\draw [fill] (0,0) circle (0.1) node[above,text width=2cm,align=center]{\small{10}};
\draw [fill] (3,-2) circle (0.1) node[right,align=center]{\small{\{0,1\}0}};
\draw[-Latex,thick] (0,0) to (3,-2);
\draw[-Latex,thick] (0,0) to (-3,-2);
\draw [fill] (-2,-3) circle (0.1)  node[right,align=center]{\small{\{0,1\}1}};
\draw [fill] (2,-3) circle (0.1) node[left,,text width=0.8cm,align=center]{\small{0\{0,1\}}\\[3mm]\phantom{a}};
\draw[-Latex,thick] (-3,-2) to (-5,-4);
\draw[-Latex,thick] (-3,-2) to (-1,-4);
\draw[-Latex,thick] (3,-2) to (5,-4);
\draw[-Latex,thick] (3,-2) to (1,-4);
\draw[-Latex,thick] (-2,-3) to (-3,-5);
\draw [fill] (-5,-4) circle (0.1) node[below,text width=1cm,align=center] {\\10\\$x_{e,1}$};
\draw [fill] (-1,-4) circle (0.1) node[below,text width=1cm,align=center] {\\01\\$x_{e,3}$};
\draw [fill] (1,-4) circle (0.1) node[below,text width=1cm,align=center] {\\01\\$x_{e,4}$};
\draw [fill] (5,-4) circle (0.1) node[below,text width=1cm,align=center] {\\10\\$x_{e,6}$};
\draw [fill] (-3,-5) circle (0.1) node[below,text width=1cm,align=center] {\\11\\$x_{e,2}$};
\draw [fill] (3,-5) circle (0.1) node[below,text width=1cm,align=center] {\\00\\$x_{e,5}$};
\draw[-Latex,thick] (2,-3) to (3,-5);
\node[draw,text width=1.9cm] at (-4,0.5)  {Switching 4\\4 mutations};
\end{scope}
\end{tikzpicture}
\caption{The four switchings possible for $N_e$. The interior nodes are labelled by the output of the bottom-up phase of Fitch's algorithm, for the two characters concerned. The $\cup$ symbol denotes where union events occur (i.e. mutations are incurred). The critical point is that both switching 2 and 4 incur the fewest
number of mutations, and these select for $01$ and $10$ at the root, respectively, representing the choice of which way to orient edge $e$.}
\label{fig:stevenfig2}
\end{figure}
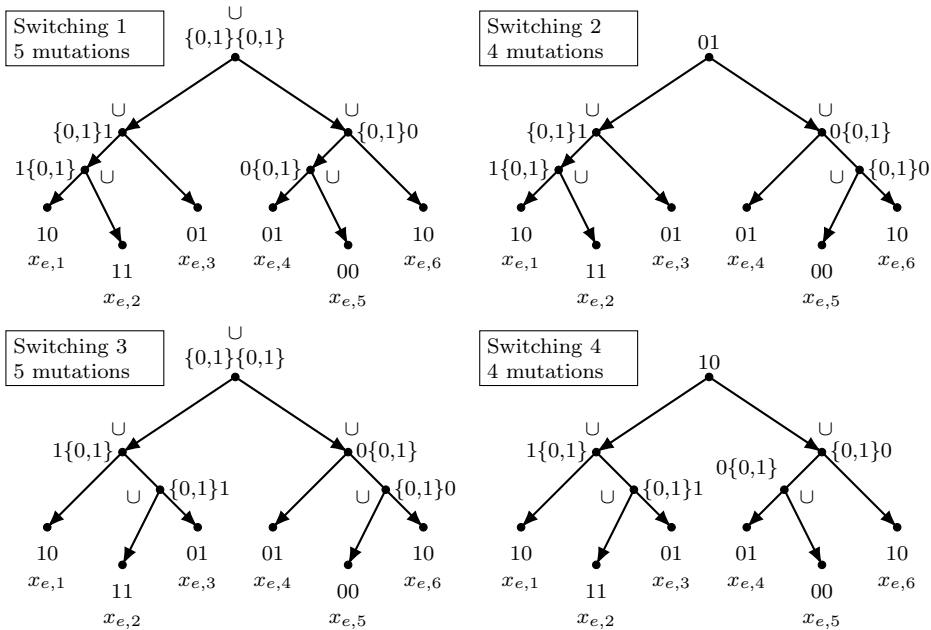

Given that each $N_e$ contains 2 reticulations, there are $2^2 = 4$ different switchings of these reticulations possible, shown in Figure 2. Note that switchings 1 and 3 both induce 5 mutations, while switchings 2 and 4 both induce 4 mutations. (Here by ``induce mutations'' we are referring to properties (i) and (ii) of Fitch's algorithm, described in the appendix). We now claim that there exists an optimum solution in which only switchings 2 and 4 are used. Suppose, for some $e = \{u,v\} \in E$, switching
1 or 3 is used.  Let $T$ be the tree induced by this switching. Fix any optimal extension of $A$ to $T$.  Let $T_e$ be the subtree of $T$ rooted at $r_e$; at least 5 mutations will be incurred on the edges of $T_e$ (with respect to the extension; see property (i) of Fitch's algorithm). Consider now the states allocated to $r_e$ in
columns $u$ and $v$. There are four such $uv$ combinations: 00, 01, 10, 11. If it is combination 01 or 10, we could replace $T_e$ with the subtree corresponding to switching 2 or 4 (respectively). This replacement subtree incurs only 4 mutations on its edges, so the total number of mutations in $T$ decreases. If it is combination 00 or 11 we can use switching 2. This might induce a new mutation (on the edge incoming to $r_e$) but we again save at least one mutation on the edges of the subtree (because at most 4, rather than at least 5 mutations are incurred there), so the overall number of mutations does not increase. Summarizing, whichever combination 00, 01, 10, 11 occurs at $r_e$, we can replace it with switching 2 or 4 without increasing the total number of mutations. Iterating this procedure proves the claim. Henceforth we can thus assume that for each $e \in E$ either switching 2 or 4 is used.

Observe that if, for a given $e = \{u,v\}$, the network $N_e$ uses switching 2, the bottom-up phase of Fitch's algorithm will allocate 01 (in columns $u$ and $v$) to $r_e$. If, on the other hand, switching 4 is used, Fitch's algorithm will allocate 10. In both cases, exactly 4 union events are generated on the nodes (of the subtree of $N_e$ induced by the switching). See Figure 2 for elucidation. 

The central idea is that, since, for an edge $e=\{u,v\}$,  a state 0 (resp. 1)  in $v$ implies a state 1 (resp. 0) in $u$ and vice versa, we can use the choice of whether to use switching 2 or 4 (for each of the $|E|$ reticulation pairs) to encode a choice as to which way to orient the corresponding edge. Without loss of generality we use state 0 to denote incoming edges, and state 1 to denote outgoing edges. Consider the bottom-up phase of Fitch's algorithm. Observe that, if a vertex $v$ incident to three edges $e_1, e_2, e_3$ becomes a \emph{sink}, the states at the roots of $N_{e_1}, N_{e_2}, N_{e_3}$ (in column $v$) will
\emph{all} be $0$, and for each $e' \not \in \{e_1, e_2, e_3\}$ the states at the
root of $N_{e'}$ (in column $v$) will be ``-'' i.e. ``don't care''\footnote{Fitch's algorithm is not well-defined on ``-'' symbols, but the intuition is that it behaves exactly like the subset of states $\{0,1\}$ behaves in Fitch's algorithm. This, in fact, is exactly how the construction described in Observation \ref{obs:recode} removes ``-'' symbols  from the alignment.}. Continuing Fitch's algorithm along the backbone of the caterpilllar shows that no mutations 
will be incurred on the edges of the caterpillar in column $v$.
A completely symmetrical situation holds if a vertex becomes a \emph{source}: 
the states at the roots of $N_{e_1}, N_{e_2}, N_{e_3}$ (in column $v$) will \emph{all} be 1, and again
no mutations are incurred on the edges of the caterpillar. On the other hand, if a vertex $v$ is neither
a source nor a sink, then the states assigned by the bottom-up phase of Fitch's algorithm to the roots of $N_{e_1}, N_{e_2}, N_{e_3}$ (in column $v$) will consist of $0$ (twice) and $1$ (once) or
$1$ (twice) and $0$ (once). Either way exactly 1 mutation is then incurred on the
edges of the caterpillar (as can be observed by running the top-down phase of Fitch's algorithm).

This means that the parsimony score is minimized by creating as many sources and sinks as possible. Specifically we have
\[
l_A(N) = 6|V| + (|V| - msso(G)).
\]
Each edge in the graph will induce 4 mutations (within the $N_e$ part),
and $|E| = 3|V|/2$, which explains the term $6|V|$. As argued above, sources and sinks to do not increase the parsimony score, and all other
vertices increase the parsimony score by exactly 1, hence the term $(|V| - msso(G))$.

Clearly $msso(G)$ can easily be calculated from $l_A(N)$. Finally, we can apply Observation
\ref{obs:recode} to obtain a network $N'$ and $A'$ without ``-'' symbols such that
\[
l_{A'}(N') = 6|V|(|E|-3) + 6|V| + (|V| - msso(G))
\]
The transformation does not raise the level of the network or the number of arcs outgoing from any biconnected component. NP-hardness follows. \qed
\end{proof}


\noindent
If we \emph{do} allow ``-'' symbols then the following slightly stronger
result is obtained: APX-hardness implies NP-hardness but additionally excludes the existence of a  Polynomial Time Approximation Scheme (PTAS), unless P=NP. APX-hardness does not obviously hold if we encode the gap symbols using Observation \ref{obs:recode} because the additive $O(|V||E|)$ term thus created distorts the objective function.

\begin{corollary}
It is APX-hard to compute the most parsimonious tree displayed by a rooted phylogenetic network $N$ with respect to an alignment $A$, even when $N$ is a binary level-1 network with at most 3 outgoing arcs per gall and $A$ consists only of states $\{0,1,\text{``-''}\}$.
\end{corollary}
\begin{proof}
We give a $(14,1)$ L-reduction from $msso$, which is APX-hard, to the parsimony
problem. L-reductions preserve APX-hardness so the result will follow. An $(\alpha, \beta)$ L-reduction (\cite{papayanna}), where $\alpha, \beta \geq 0$, is defined as follows.

\begin{definition}
Let $A,B$ be two optimization problems and $c_A$ and $c_B$ their respective cost functions. A pair of functions $f,g$, both computable in polynomial time, constitute an $(\alpha, \beta)$
L-reduction from $A$ to $B$ if the following conditions are true:

\begin{enumerate}
\item For every instance $x$ of $A$, $f(x)$ is an instance of $B$,
\item For every feasible solution $y$ of $f(x)$, $g(y)$ is a feasible solution of $x$,
\item For every instance $x$ of $A$, $OPT_B(f(x)) \leq \alpha OPT_A(x)$,
\item For every feasible solution $y'$ of $f(x)$ we have $|OPT_A(x) - c_A(g(y'))| \leq \beta |OPT_B(f(x)) - c_B(y')|$
\end{enumerate}
where $OPT_{A}$ is the optimal solution value of problem $A$ and similarly for $B$.
\end{definition}

 For brevity we refer to the optimum size of the parsimony problem as $mp(N,A)$. We use the reduction described in the proof of 
Theorem \ref{thm:NPhardness_MP} (before the gap symbols have been
removed) with some slight modifications. The forward-mapping function $f$ (condition 1 of the L-reduction) is the same mapping used in the proof of Theorem \ref{thm:NPhardness_MP}. The back-mapping function $g$ (i.e. condition 2) will be described below. To establish condition 3 for a given $(\alpha, \beta)$ we need to prove that $mp(N,A) \leq \alpha \cdot msso(G)$. Now,
we know that $msso(G) = maxcut(G)-|V|/2$ (see appendix) and that $maxcut(G) \geq 2/3 |E| = |V|$ (because
every cubic graph has a cut at least this large simply by moving nodes which have more
neighbours on their side of the cut, to the other side). Hence, $msso(G) \geq |V|/2$. We
know that $mp(N,A) = 7|V| - msso(G)$. Trivially therefore $mp(N,A) \leq 7|V|$. Hence
taking $\alpha=14$ is sufficient. For the other direction, we need to show that for
an arbitrary solution to the parsimony problem, which induces $p$ mutations, the
back-mapping function yields an orientation of $G$ with $s$ sources and sinks such that
$|msso(G)-s| \leq \beta |p-mp(N,A)|$. The back-mapping function $g$ first ensures that
all the $N_e$ gadgets are using type 2 or type 4 switchings, which might reduce
the number of mutations to $p' \leq p$, and then extracts an orientation of $G$ (thus establishing condition 2).  Now, $s = 7|V|-p'$ and $mp(N,A) = 7|V|-msso(G)$ so
\begin{align*}
msso(G) - s &= msso(G) - (7|V|-p')\\
&= msso(G) - 7|V| + p'\\
&= p' - (7|V| - msso(G))\\
& = p' - mp(N,A)\\
& \leq p - mp(N,A).
\end{align*}
So taking $\beta = 1$ is sufficient to establish condition 4. \qed
\end{proof}

\section{Hardness of finding the most likely tree displayed by a network}
\label{sec:MLtree}

\textbf{Preliminaries on the likelihood of a tree}.
We now introduce the basic concepts and notation that are necessary to 
define the likelihood of a tree with respect to an alignment. First, 
we need a probabilistic model describing how sequences evolve along a tree.
Here we assume the simplest model available, 
known as the Cavender-Farris model (\cite{farris1973,cavender1978}), 
which can be described as follows.
Let $T = (V,E)$ be a rooted binary phylogenetic tree on $X$. 
We associate probabilities $\mathbf{p}=(p_e)_{e\in E} \in [0,\nicefrac{1}{2}]^{|E|}$
to the edges of $T$ and denote this $(T,\mathbf{p})$. 
Under the Cavender-Farris model, each character evolves independently, as follows: 
at the root pick randomly a state between 0 and 1, each with probability $\nicefrac{1}{2}$, 
and then, for each vertex $v$ below the root, either copy the state of the parent of $v$ or flip it,
with probabilities $1-p_e$ and $p_e$, respectively.
The restriction $p_e<\nicefrac{1}{2}$ corresponds to the fact
that, in a symmetric model, no amount of time can make a character 
more likely to change state than to remain in the same state.

The process described above eventually associates
a state to each element of $X$ at the leaves of the tree, that is, it generates a 
random binary character.
The probability of generating the binary character $f$ is called the \emph{likelihood} of $(T,\mathbf{p})$
with respect to $f$, denoted $L_f(T,\mathbf{p})$, and can be calculated as follows:
\[
L_f(T,\mathbf{p}) =  \sum_{\hat{f}} \frac{1}{2} \prod_{e = (u,v) \in E} p_e^{|\hat{f}(v)-\hat{f}(u)|} (1-p_e)^{1-|\hat{f}(v)-\hat{f}(u)|}
\]
Here $\hat{f}$ ranges over all extensions of $f$ to $T$. 
Because the model assumes that the characters in a sequence evolve independently,
the probability of generating the binary sequences in an alignment $A$, named the \emph{likelihood} of $(T,\mathbf{p})$
with respect to $A$, denoted $L_A(T,p)$, can be obtained as
\[
L_A(T,\mathbf{p}) = \prod_{f \in A} L_f(T,\mathbf{p})
\]
(Here, and in the rest of this section, we assume that alignments do not contain gap symbols.)

We now introduce some more notation that will be useful in the following. 
An \emph{extension} $\widehat{A}$ of an alignment $A$ to a tree $T=(V,E)$ is a 
set of functions $\hat{f}: V \rightarrow \{0,1\}$ obtained by taking exactly one
extension of each character in $A$. In practice, $\widehat{A}$ can be 
represented as a matrix with $|V|$ rows and $|A|$ columns, in which 
the rows corresponding to the leaves of $T$ are identical to the rows of $A$.
For $e=(u,v) \in E$, we denote by $h_e(\widehat{A})$ the
number of differences (that is, the Hamming distance) between the sequences
that $\widehat{A}$ associates to $u$ and $v$.
Finally, let $l_{\widehat{A}}(T)$ denote 
$\sum_{e\in E} h_{e}(\widehat{A}) = \sum_{\hat{f}\in \widehat{A}} l_{\hat{f}}(T)$.  
Note that the parsimony score $l_{A}(T)$ is the minimum of $l_{\widehat{A}}(T)$
over all extensions of $A$. Given these notations, we can express the likelihood
of $(T,\mathbf{p})$ as follows, where $m=|A|=|\widehat{A}|$, and $\widehat{A}$
ranges over all extensions of $A$:
\begin{equation}
L_A(T,\mathbf{p}) = \sum_{\widehat{A}} 2^{-m}\prod_{e \in E} p_e^{h_{e}(\widehat{A})} (1-p_e)^{m-h_{e}(\widehat{A})} \label{eqn:lik2}
\end{equation}

\textbf{Trees displayed by a network with edge probabilities}.
The notation above can be extended to networks with probabilities
$p_e'\in [0,\nicefrac{1}{2}]$
assigned to their edges, denoted $(N,\mathbf{p'})$. 
We say that a network $(N,\mathbf{p'})$ displays a tree $(T,\mathbf{p})$, 
if $N$ displays $T$ in the usual topological sense 
(i.e. some subdivision $T'$ of $T$ is a subtree of $N$) 
\emph{and} $(T,\mathbf{p})$ can be obtained from $T'$ by repeatedly suppressing
vertices with indegree-1 and outdegree-1, where here suppression
also updates the probabilities. Specifically, if $T'$ contains
two edges $e_1 = (u,v)$ and $e_2 = (v,w)$, where $v$ has indegree-1
and outdegree-1, the suppression operation replaces these two
edges with a single edge $e = (u,w)$ and assigns it the probability
\[
p_e = p_{e_1}(1-p_{e_2}) + (1-p_{e_1})p_{e_2}.
\]
This expresses the probability of having different 
states at the endpoints of a two-edge path, under the Cavender-Farris model.
Note that, in general, a tree $T$ can have multiple distinct images $T'$ 
in the network, so it can occur that $(N,\mathbf{p'})$ displays $(T,\mathbf{p})$ 
for multiple different $\mathbf{p}$. 
Also note that because $p_e'\le\nicefrac{1}{2}$ for all edges in the network,
the same will hold for the edges of the trees it displays, 
as no application of the equation above can produce a probability
$p_e>\nicefrac{1}{2}$ from edge probabilities that are at most 
$\nicefrac{1}{2}$. It is also easy to see that if $0< p_{e_1}, p_{e_2}< \nicefrac{1}{2}$, 
then $\max\{p_{e_1},p_{e_2}\} < p_e < p_{e_1}+p_{e_2}$. 
These observations lead to the following one, which will be useful later on:

\begin{observation} \label{obs:bounds-p_e}
Let $(N,\mathbf{p'})$ be such that for every edge of $N$, $0<p'_e<\nicefrac{1}{2}$.
Let $(T,\mathbf{p})$ be a tree displayed by $(N,\mathbf{p'})$ and $e$ an edge of $T$.
Finally, let $E'(e)$ be the subset of the edges of $N$ 
whose probabilities contribute to $p_e$.
Then, $p_e < \nicefrac{1}{2}$ and
\[
\max_{e'\in E'(e)} p_{e'} < p_e < \sum_{e'\in E'(e)} p_{e'}
\]
\end{observation}


We say that $(T^*,\mathbf{p^*})$ is the/a \emph{most likely (ML) 
tree displayed by $(N,\mathbf{p'})$ (with respect to $A$)} if it 
maximizes $L_A(T,\mathbf{p})$, ranging over all $(T, \mathbf{p})$ displayed 
by $(N,\mathbf{p'})$. 
In the remainder of this section we consider the problem 
of finding such a most likely tree given a network 
with edge probabilities and an alignment.

\textbf{A link between likelihood and parsimony}.
There are well-known relationships between the likelihood and the 
parsimony of a tree that imply that under some conditions the most 
likely tree is also a most parsimonious one (\cite{tuffley1997}).
We now illustrate one such relationship (Corollary \ref{cor:ML-MP} below), which is based on the 
observation that as we reduce the scale of a tree, its likelihood converges 
to zero at a rate that only depends on  its parsimony score.
Although it shares similarities with the results by Tuffley and Steel,
we are not aware that it has been explicitly stated in the literature.
This result is \emph{not} necessary to obtain the other results in this section, 
but it provides the intuition behind them.

In the following statements, we assume that $c\in]0,1]$, so the form 
$c\rightarrow 0$ is to be understood as $c$ approaches 0 to the right.
Also, $c\mathbf{p}$ simply denotes the product between the scalar
$c$ and vector $\mathbf{p}$.

\begin{lemma}
The function $f(c) = L_A(T,c\mathbf{p})$ is $\Theta(c^{l_A(T)})$ as $c\rightarrow 0$.
\end{lemma}
\begin{proof}
Write $L_A(T,c\mathbf{p})$ using Eqn.~(\ref{eqn:lik2}):
\[
L_A(T,c\mathbf{p}) = \sum_{\widehat{A}} 2^{-m}c^{l_{\widehat{A}}(T)}\cdot 
\prod_{e \in E} p_e^{h_{e}(\widehat{A})} (1-cp_e)^{m-h_{e}(\widehat{A})}.
\]
Note that the products above tend to a constant as $c\rightarrow 0$. 
As a consequence, the term for $\widehat{A}$ in the sum 
has order $\Theta(c^{l_{\widehat{A}}(T)})$ as $c\rightarrow 0$. 
Since the lowest degree dominates, their sum is $\Theta(c^{l_A(T)})$. \qed
\end{proof}

\begin{corollary}\label{cor:ML-MP}
Let $A$ be an alignment and $T_1$ and $T_2$ two trees such that 
$l_A(T_1) < l_A(T_2)$.
Then, for any $\mathbf{p_1}$ and $\mathbf{p_2}$,
\[
L_A(T_1,c\mathbf{p_1}) > L_A(T_2,c\mathbf{p_2})\quad\text{for $c$ sufficiently close to 0.}
\]
\end{corollary}
\begin{proof}
As $c\rightarrow 0$, $L_A(T_1,c\mathbf{p_1})$ converges to 0 at a lower rate
than $L_A(T_2,c\mathbf{p_2})$. Thus there exists a neighborhood of 0 in which the desired inequality holds. \qed
\end{proof}

The corollary above can be extended to any collection of trees: irrespective of the
edge probabilities assigned to them, if the trees are rescaled by a sufficiently
small $c$, the most parsimonious trees will have likelihoods greater than all 
the other trees, meaning that a most likely tree in the collection of rescaled
trees will necessarily also be most parsimonious. 





\textbf{Proving the NP-hardness of finding an ML tree in a network}.
In the remainder of this section, namely in the statements of the next
two formal results, we are implicitly given a network $N$ on $X$ with $|X|=n$
and an alignment $A$ with $m$ characters on $X$.
The \emph{height} of a network $N$ is the maximum number of edges in a
directed path in $N$.

\begin{lemma}\label{lem:bounds}
Let $(N,\mathbf{c})$ be a network of height $d_N$,  
where all the edges are assigned a constant probability $c_e=c$,
with $0<c<\nicefrac{1}{2}$.
Let $(T,\mathbf{p})$ be a tree displayed by $(N,\mathbf{c})$. Then,
\[
2^{-2mn} \cdot c^{l_A(T)} < L_A(T,\mathbf{p}) < 
2^{mn} \cdot d_N^{2mn} \cdot c^{l_A(T)}
\]
\end{lemma}
\begin{proof}
Using Observation \ref{obs:bounds-p_e}, we note that,
for any edge $e\in E$ of $T=(V,E)$, $c<p_e<\min\{cd_N,\nicefrac{1}{2}\}$.

We begin by proving the upper bound in the statement.
From Eqn.~(\ref{eqn:lik2}) and the fact that $(1-p_e)^{m-h_{e}(\widehat{A})} <1$,
we get the first inequality in the following:
\begin{align*}
L_A(T,\mathbf{p}) &< \sum_{\widehat{A}} 2^{-m}\prod_{e \in E} (cd_N)^{h_{e}(\widehat{A})} \\
&= \sum_{\widehat{A}} 2^{-m} d_N^{l_{\widehat{A}}(T)} c^{l_{\widehat{A}}(T)} 
\quad < 2^{m(n-1)-m} d_N^{m(2n-2)} c^{l_A(T)},
\end{align*}
where the last inequality is obtained by noting that the sum has $2^{m(n-1)}$ terms (there are $n-1$ internal nodes in a rooted binary tree, and thus $2^{m(n-1)}$ different extensions of $A$), and that $l_A(T) \le l_{\widehat{A}}(T) \le m(2n-2)$  (there are $2n-2$ branches in a rooted binary tree, and thus we cannot have more than $2n-2$ changes per character).
The upper bound in the statement is larger than the one above.

As for the lower bound, if we use $c<p_e<\nicefrac{1}{2}$ 
in Eqn.~(\ref{eqn:lik2}):
\begin{align*}
L_A(T,\mathbf{p}) &> \sum_{\widehat{A}} 2^{-m}\prod_{e \in E} c^{h_{e}(\widehat{A})} \nicefrac{1}{2}^m\\
&= \sum_{\widehat{A}} 2^{-m} c^{l_{\widehat{A}}(T)} 2^{-m(2n-2)} \quad > 2^{-m(2n-1)} c^{l_A(T)},
\end{align*}
where the last inequality is obtained by taking only one term in the sum. 
The lower bound in the statement is smaller than the one above. \qed
\end{proof}

The lemma above shows the order of convergence to 0 of the likelihood $L_A(T,\mathbf{p})$
of a tree displayed by $(N,\mathbf{c})$ as  $c\rightarrow 0$.
The higher the parsimony score, the faster the convergence.
As a consequence, for $c$ sufficiently close to 0, a tree with a lower parsimony 
score than another will have a higher likelihood.
The following lemma shows how close is ``sufficiently close'', by providing an
explicit upper bound to $c$.

\begin{proposition}
Let $(N,\mathbf{c})$ be a network of height $d_N$,  
where all the edges are assigned a constant probability $c_e=c$,
with $0<c<d_N^{-2mn}2^{-3mn}$.
If $(T^*,\mathbf{p^*})$ is a most likely tree displayed by $(N,\mathbf{c})$,
then $T^*$ is a most parsimonious tree displayed by $N$.
\end{proposition}
\begin{proof}
Suppose that $(T^*,\mathbf{p^*})$ is a most likely tree displayed by $(N,\mathbf{c})$, but not most parsimonious.
That is, there exists $(T,\mathbf{p})$ displayed by $(N,\mathbf{c})$ with $l_A(T)\le l_A(T^*)-1$.
But then, by using the lower bound in Lemma \ref{lem:bounds}:
\[
L_A(T,\mathbf{p}) > 2^{-2mn} \cdot c^{l_A(T)} \ge 2^{-2mn} \cdot c^{l_A(T^*)-1}.
\]
Now apply the upper bound in Lemma \ref{lem:bounds} to $T^*$, and combine it with $c<d_N^{-2mn}2^{-3mn}$:
\[
L_A(T^*,\mathbf{p^*}) 
< 2^{mn} \cdot d_N^{2mn} \cdot c^{l_A(T^*)} 
< 2^{mn} \cdot d_N^{2mn} \cdot c^{l_A(T^*)-1} \cdot d_N^{-2mn} \cdot 2^{-3mn} = 2^{-2mn} \cdot c^{l_A(T^*)-1}
\]
The last terms of the two chains of inequalities above are equal, thus proving $L_A(T,\mathbf{p})>L_A(T^*,\mathbf{p^*})$.
Since this contradicts the assumption that $(T^*,\mathbf{p^*})$ is a most likely tree, the statement follows. \qed
\end{proof}

The proposition above shows that the NP-hard problem of finding the/a most 
parsimonious tree in a network $N$ with respect to an alignment $A$
can be reduced to the problem of finding the/a most 
likely tree in $(N,\mathbf{c})$ with respect to $A$,
where $\mathbf{c}=(c_e)$ is such that $c_e=c$,
and $0<c<d_N^{-2mn}2^{-3mn}$.
Since the reduction preserves the network and the alignment, 
the main result of this section follows from Theorem \ref{thm:NPhardness_MP}:

\begin{theorem} 
It is NP-hard to compute the most likely tree $(T,\mathbf{p})$ displayed by a 
rooted phylogenetic network $(N,\mathbf{p'})$ with respect to an alignment $A$, 
even when $N$ is a binary level-1 network with at most 3 outgoing arcs per gall 
and $A$ consists only of two states $\{0,1\}$ 
and does not contain gap symbols.
\end{theorem}

\section{Conclusions and open problems}


We have shown that, given a phylogenetic network with a sequence for each leaf, finding a most parsimonious or most likely tree displayed by the network is computationally intractable (NP-hard). Moreover, this is the case even when we restrict to binary sequences and level-1 networks; the simplest networks that are not trees. However, many computational problems that can be shown to be theoretically intractable can be solved reasonably efficiently in practice (see e.g. \emph{Cautionary Tales of Inapproximability} by \cite{budden2017cautionary}). We end the paper by discussing whether we expect this to be the case for our problem.

There is a dynamic programming algorithm, described in Theorem 5.7 of~\cite{fischer2015computing}, for finding a tree in a network that is most parsimonious with respect to a single character. The running time is fixed-parameter tractable, with as parameter the level of the network. Hence, this algorithm is practical as long as the level of the network is not too large. This algorithm can easily be extended to multiple characters (that all have to choose the same tree) when the number of characters is adopted as a second parameter. Indeed, for every root of a biconnected component, we introduce a dynamic programming entry not just for every possible state but for every possible sequence of states. However, the running time of this algorithm would be exponential in the number of characters, which makes it useless for almost all biological data. Similarly, the Integer Linear Programming (ILP) solution presented in the same paper can also be easily extended to multiple characters. However, there does not seem to be an easy way to do that without having the number of variables growing linearly in the number of characters. Hence, this approach is also unlikely to be useful in practice.

In contrast, consider the simple algorithm that loops through the at most $2^r$ trees displayed by the network, with~$r$ the number of reticulation nodes in the network, and computes the parsimony or likelihood of each tree (this na\"ive FPT algorithm was presented in \cite{nakhleh2005reconstructing}, where it is named \emph{Net2Trees}). Ironically, this simple algorithm would outperform the approaches mentioned above for any kind of data with a reasonably large number of characters. Hence, the main open question that remains is whether there exists an algorithm whose running time is linear (or at least polynomial) in the number of characters and whose dependency on~$r$ is better than~$2^r$ (for example recently an algorithm with exponential base smaller than 2 was discovered for the tree containment problem (\cite{gunawan2016program}), although this algorithm does not obviously extend to generating all trees in the network). Another question of  interest that remains open is the following: does the parsimony problem under restrictions (1)-(3) listed in the introduction permit good (i.e. constant factor) approximation algorithms, and possibly even a PTAS, when the alignment $A$ does not contain any indels?

\bibliographystyle{spbasic}
\bibliography{findingParsimony.bib}

\section*{Acknowledgments}
Leo van Iersel was partly supported by the Netherlands Organization for Scientific Research (NWO), including Vidi grant 639.072.602, and partly by the 4TU Applied Mathematics Institute. Celine Scornavacca was partly supported by  the French Agence Nationale de la Recherche Investissements d'Avenir/Bioinformatique (ANR-10-BINF-01-02, Ancestrome).
\pagebreak 

\appendix

\section{Appendix: Fitch's algorithm}

Fitch's algorithm \cite{fitch1971} has two phases. In the first phase, known as the \emph{bottom-up} phase, we start by
assigning the singleton subset of states $\{f(x)\}$ to each taxon $x$. The internal nodes of $T$
are assigned subsets of states recursively, as follows. Suppose a node $p$ has two children
$u$ and $v$, and the bottom-up phase has already assigned subsets $F(u)$ and $F(v)$ to the two children, respectively. If $F(u) \cap F(v) \neq \emptyset$ then set $F(p) = F(u) \cap F(v)$ (in which case we say that $p$ is an \emph{intersection} node). If $F(u) \cap F(v) = \emptyset$ then set $F(p) = F(u) \cup F(v)$ (in which case we say that $p$ is a \emph{union} node). The number of union nodes in the bottom-up phase is equal to $l_f(T)$. To actually create an optimal extension $\hat{f}$, we require the \emph{top-down} phase of Fitch's algorithm. Start at the root $r$ and let $\hat{f}(r)$ be any element in
$F(r)$. For an internal node $u$ with parent $p$, we set $\hat{f}(u) = \hat{f}(p)$ (if $\hat{f}(p) \in F(u)$) and otherwise (i.e. $\hat{f}(p) \not \in F(u)$) set $\hat{f}(u)$ to
be an arbitrary element of $F(u)$.\\
\\
For each node $u$ of the tree, let $\cup(u)$ be the number of union events in the subtree rooted at $u$. The following well-known properties of Fitch's algorithm are used repeatedly in the main hardness proof of this article: (i) \emph{every} extension (optimal or otherwise) must incur
\emph{at least} $\cup(u)$ mutations on the edges of the
subtree rooted at $u$; (ii) an extension created
by Fitch's algorithm induces \emph{exactly} $\cup(u)$ mutations
on the edges of the subtree rooted at $u$ (and
$u$ is assigned a state from $F(u)$ in this
extension).

\section{Appendix: NP-hardness and APX-hardness of MAX-SOURCE-SINKS-ORIENTATION}

The following result is based on a sketch proof by Colin McQuillan\footnote{TCS Stack Exchange, 2010, URL: \url{http://cstheory.stackexchange.com/questions/2307/an-edge-partitioning-problem-on-cubic-graphs/}}. We have been unable to find an original reference and hence have reconstructed the proof in detail. The APX-hardness proof is original.\\

\noindent
\textbf{Lemma \ref{lem:sourcesinkhard}}. \emph{MAX-SOURCE-SINKS-ORIENTATION is NP-hard on cubic graphs.}
\begin{proof}
Recall that the classical MAX-CUT problem asks us to bipartition the vertex set of an unidrected graph $G$, such that the number of edges that cross the bipartition is maximized. We reduce from the NP-hard problem CUBIC-MAX-CUT which is the restriction of the MAX-CUT problem to cubic graphs. Given an undirected cubic graph $G$, we simply write $maxcut(G)$ to denote the number of edges in the maximum-size cut.

We reduce CUBIC-MAX-CUT to MAX-SOURCE-SINKS-ORIENTATION. Specfically, given an undirected cubic graph $G = (V,E)$ (i.e. an instance of CUBIC-MAX-CUT) we will show that $msso(G) =maxcut(G) - |V|/2$, from which the hardness will follow.

We start by proving that $msso(G) \geq maxcut(G) - |V|/2$. Fix \steven{an arbitrary cut $C$} of $G$ and let $(U,W)$ be the corresponding bipartition. If some vertex of $U$ or
$W$ has more neighbours on the other side of the partition than its own, move it to the other
side of the partition: this will increase the size of the cut. We repeat this until it is no
longer possible and let $C$ and $(U, W)$ refer to the cut and its induced partition at the end of this process. Note that now each vertex in $U$ (respectively, $W$) will have at most one neighbour in $U$ (respectively, $W$). We proceed by orienting the
edges in the cut from $U$ to $W$. Now, the remaining edges are either internal to $U$ or internal to $W$. These edges must form a matching (i.e. they are node disjoint). For each such edge in $U$ (respectively, $W$), exactly one endpoint will become a source (respectively, sink). Nodes in $U$ (respectively, $W$) that are not adjacent to internal edges will already be sources (respectively, sinks) due to the orientation of the cut edges from $U$ to $W$. Hence, if
we write $|C|$ to denote the number of edges in the cut $C$, we obtain an orientation
of $G$ with at least
\[
(|E| -|C|) + (|V| - 2(|E|-|C|))
\]
sources and sinks. Hence,
\begin{align*}
msso(G) &\geq (|E| -maxcut(G)) + (|V| - 2(|E|-maxcut(G)))\\
 & = maxcut(G) + |V| - |E| \\
& = maxcut(G) + |V| - (3/2)|V| \\
& = maxcut(G) - |V|/2. 
\end{align*}

For the other direction, \steven{fix an arbitrary orientation of $G$ and let $s$ be the number of sources and sinks
created by the orientation}. We write $V_i$ ($i \in \{0,1,2,3\}$) to denote those vertices of $G$ which
have indegree $i$. Let $U = V_0 \cup V_1$ and let $W = V_2 \cup V_3$. Whenever
an edge of $G$ has been oriented from $W$ to $U$, reverse its
orientation: this only decreases the indegrees of the vertices in $U$ and increases the
indegrees of vertices in $W$ so it cannot destroy any sources or sinks and it cannot cause a node to be on the ``wrong'' side of the bipartition. \steven{(In fact, it will cause the number of
sources and sinks to increase, so this situation can only occur if the orientation was not
optimal). Let $s$ now refer to the number of sources and sinks once all arcs have been
oriented from $U$ to $W$.} \steven{The} edges $(u,w)$ such that $u \in U$ and $w \in W$ form a cut; it remains only to count how many of these edges there are. We first count
from the perspective of the vertices in $U$. The nodes in $V_0$ each generate 3 outgoing
arcs. Let $k$ be the total number of edges of the form $(u_0, u_1)$ where $u_0 \in V_0$
and $u_1 \in V_1$. Note that each node in $V_1$ that does not receive any of these
$k$ arcs, must receive an arc which is outgoing from some other node in $V_1$. 
It follows that the number of edges in the cut is
\[
(3|V_0| - k) + (2|V_1| - (|V_1|-k)) = 3|V_0| + |V_1|.
\]
If we count in a symmetrical fashion from the perspective of $W$, and let $\ell$ be the number
of arcs of the form $(u_2, u_3)$ where $u_2 \in V_2$ and $u_3 \in V_3$, it follows that the
number of edges in the cut is
\[
(3|V_3| - \ell) + (2|V_2| - (|V_2| - \ell)) = 3|V_3| + |V_2|.
\]
If we sum these two equations,  we obtain \steven{a cut with at least the following number
of edges:}
\begin{align*}
&\geq (3/2)( |V_0| + |V_3| ) + (1/2)( |V_1| + |V_2| ) \\
&= (3/2) s + (1/2)(|V| - s)\\ 
&=  s + (1/2)|V|.
\end{align*}
From this follows that $msso(G) \leq maxcut(G) - (1/2)|V|$. 
\end{proof}

The hardness of $msso$ can be strengthened to the following inapproximability result. Note
that one consequence of APX-hardness is that $msso$ does not permit a PTAS, unless $P=NP$.\\
\\ 
\noindent
\textbf{Corollary \ref{cor:sourcesinksapxhard}}. \emph{MAX-SOURCE-SINKS-ORIENTATION is APX-hard on cubic graphs.}
\begin{proof}
Note that the constructions and transformations used in the proof of  Lemma \ref{lem:sourcesinkhard} are all constructive and can easily be conducted in polynomial time.
Moreover, they apply to arbitrary cuts/orientations, and not just optimal ones. This allows
us to easily strengthen the described reduction to obtain a $(1, 1)$ L-reduction from 
 CUBIC-MAX-CUT to MAX-SOURCE-SINKS-ORIENTATION (see the main text for the definition of L-reduction). From this APX-hardness will
follow, since CUBIC-MAX-CUT is APX-hard (\cite{alimontikann,bermankarpinski}) and L-reductions are APX-hardness preserving. The
$(1,1)$ means that the inapproximability threshhold for MAX-SOURCE-SINKS-ORIENTATION
is at least as strong as that for CUBIC-MAX-CUT.

Let $G = (V,E)$ be an instance of CUBIC-MAX-CUT. The forward mapping function $f$ (from instances of CUBIC-MAX-CUT to MAX-SOURCE-SINKS-ORIENTATION) is simply the identity function. To be an $(\alpha,\beta)$ L-reduction, where $\alpha,\beta \geq 0$, we first have to show that $msso( f(G) ) = msso( G ) \leq \alpha \cdot maxcut(G)$. We know that
$maxcut(G) - (1/2)|V| = msso(G)$, so $\alpha = 1$ is trivially satisfied. We next have to show
a polynomial-time computable back-mapping function $g$ (from feasible solutions of
 MAX-SOURCE-SINKS-ORIENTATION to feasible solutions of CUBIC-MAX-CUT) with the
following property: an orientation that induces $s$ sources and sinks is mapped to
a cut with $k$ edges such that $|maxcut(G)-k| \leq \beta|msso(G) - s|$. The
function $g$ was already implicitly described in the NP-hardness reduction: reverse any
edges oriented from $W$ to $U$ (which possibly increases the number of
sources and sinks to $s' \geq s$) and extract a cut of size $s' + |V|/2$. Observe,
\begin{align*}
maxcut(G) - (s' +|V|/2) &\leq maxcut(G) - (s +|V|/2)\\
&\leq (msso(G)+|V|/2) - (s+|V|/2)\\
&\leq msso(G) -s.
\end{align*}
So taking $\beta=1$ is sufficient.
\end{proof}



%
%

\end{document}